\documentclass[12pt]{iopart}
\usepackage[utf8]{inputenc}
\usepackage[T1]{fontenc}
\usepackage[
 backend=biber,
 style=numeric,
 sorting=none
 ]{biblatex}
 \AtEveryBibitem{\clearfield{pages}}
\AtEveryBibitem{\clearfield{volume}}
\AtEveryBibitem{\clearfield{number}}
\AtEveryBibitem{\clearfield{url}}
\AtEveryBibitem{\clearfield{issn}}
\AtEveryBibitem{\clearfield{isbn}}
\AtEveryBibitem{\clearfield{publisher}}
\AtEveryBibitem{\clearfield{note}}
\AtEveryBibitem{\ifentrytype{article}
    {\clearfield{eprint}}
    {}}
\makeatletter
\newcommand{\mainmatter}{%
  \setcounter{footnote}{0}%
  \patchcmd{\@makefntext}{\fnsymbol}{\arabic}{}{}%
  \patchcmd{\@thefnmark}{\fnsymbol}{\arabic}{}{}%
  \def\@makefnmark{\textsuperscript{\arabic{footnote}}}%
}
\makeatother
\DeclareFieldFormat[misc]{title}{``#1''} 
\usepackage{iopams}
\usepackage{capt-of}
\expandafter\let\csname equation*\endcsname\relax

\expandafter\let\csname endequation*\endcsname\relax
\usepackage{hyperref}
\usepackage{amsmath}
\usepackage{dcolumn}
\usepackage{bm}
\usepackage{algorithm}
\usepackage{algpseudocode}
\newcommand{\tinyspace}{\mspace{1mu}}

\newcommand{\triplenorm}[1]{
  \left|\!\tinyspace\left|\!\tinyspace\left| #1 
  \right|\!\tinyspace\right|\!\tinyspace\right|}
\makeatletter
\renewcommand{\ALG@name}{Protocol}
\makeatother
 \usepackage[nice]{nicefrac}
 \usepackage{tikz}
 \usepackage{mathrsfs}
 \usetikzlibrary{math}
 \usetikzlibrary{quantikz}
 \usetikzlibrary{positioning}
 \newtheorem{theorem}{Theorem}
\newtheorem{definition}{Definition}
\newtheorem{lemma}{Lemma}

\newenvironment{proof}{\begin{em}Proof:~}{\end{em}}
\newcommand{\qed}{\ensuremath{\hfill\blacksquare}}

\newcommand{\mand}{\quad\text{and}\quad}
\def\D{^{\dagger}}

\def\Cc{\mathbb{C}}

\def\id{\mathbbm{1}}
\def\R{\mathcal{R}}
\def\S{\mathcal{S}}
\def\P{\mathcal{P}}

\def\U{\mathcal{U}}
\def\X{\mathcal{X}}
\def\Y{\mathcal{Y}}

\def\Qsf{\mathsf{Q}}
\def\Ksf{\mathsf{K}}
\def\Rsf{\mathsf{R}}

\def\nnn{\nonumber \\}
\usepackage{physics}
\usepackage{bbm}
\usepackage{bm}
\DeclareMathOperator{\unitary}{U}
\DeclareMathOperator{\channel}{C}
\DeclareMathOperator{\density}{D}
\DeclareMathOperator{\sphere}{S}
\DeclareMathOperator{\pos}{Pos}

\newcommand{\proref}[1]{Protocol \ref{#1}}
\newcommand{\figref}[1]{Figure \ref{#1}}
\newcommand{\thmref}[1]{Theorem \ref{#1}}
\newcommand{\defref}[1]{Definition \ref{#1}}
\newcommand{\lemref}[1]{Lemma \ref{#1}}
\newcommand{\notimes}{\bigotimes_{\substack{i=1\\i\neq\ell}}^{n+1}}
\newcommand{\noprod}{\prod_{\substack{i=1\\i\neq \ell}}^{n+1}}
\begin{document}

\title[Verified delegated quantum computation requires techniques beyond cut-and-choose]{Verified delegated quantum computation requires techniques beyond cut-and-choose}

\author{Fabian Wiesner\textsuperscript{1,$\dagger$} and Anna Pappa\textsuperscript{1}\\[1em]
\footnotesize\textsuperscript{1}{%
Electrical Engineering and Computer Science Department,\\Technische Universit\"at Berlin, 10587 Berlin, Germany
}}%
\ead{\textsuperscript{$\dagger$}f.wiesner@tu-berlin.de}
\vspace{10pt}
\setlength{\mathindent}{0pt}

\begin{abstract}
    Delegated quantum computation enables a client with limited quantum capabilities to outsource computations to a more powerful quantum server while preserving correctness and privacy. Verification is crucial in this setting to ensure that the untrusted quantum server performs the computation honestly and returns correct results. A common verification method is the quantum cut-and-choose technique. Inspired by classical verification methods for two-party computation, the client uses the majority of the delegated rounds to test the server's honesty, while keeping the remaining ones for the actual computation. Combining this technique with other methods, such as quantum error correction, could help achieve negligible cheating probabilities for the server; however, such methods can impose significant overheads making implementations unfeasible for the near-term future. In this work, we investigate whether cut-and-choose can yield efficient and secure verifiable quantum computation without additional costly techniques. We find that verifiable delegated quantum computation protocols relying solely on cut-and-choose techniques cannot be secure and efficient at the same time. 
\end{abstract}
\section{\label{sec:Introduction}Introduction}\mainmatter
Several technological advances towards usable quantum computers have been made recently. Due to the expected unprecedented speedup for relevant computational problems, such as factoring and simulating chemical reactions, quantum computers are appealing to scientific institutions and companies alike. However, compared to classical computers, current and near-term implementations of quantum computers are expensive, massive in size, and require considerable effort to operate and maintain. Indeed, there is no indication that these drawbacks will be alleviated in the foreseeable future, if ever, and that non-specialized organizations will be able to operate quantum computers. One solution to this problem is cloud access via the (quantum) Internet. Although companies building quantum computers have already implemented this solution, it comes with significant downsides. If a client requires access to a quantum computer, beyond mere testing or playful exploration, it is reasonable to assume that the delegated computation is non-trivial. This implies that the algorithm, as well as the input and output data, may be of significant value. Consequently, if the service provider (server) is untrusted, the client cannot risk the disclosure of this sensitive information.

\emph{Blindness} is the property that ensures that no information about the computation, the input, or the output is leaked. Recent explorations \cite{Badertscher_2020} provide evidence that achieving blindness for purely classical clients is unlikely without additional assumptions, such as the quantum hardness of the Learning with Errors (LWE) problem. Remarkably, granting the client only minimal quantum capabilities, such as preparing or measuring single qubits and performing Pauli operations, is sufficient to obtain perfect, information-theoretic blindness \cite{BFK09, MF13}.

Beyond preserving privacy, the client must also be assured that the computation is performed correctly; therefore, the protocol should provide \emph{verifiability}. A widely used approach is for the client to assess the server’s honesty by delegating multiple computations: in certain rounds (test rounds), the client assigns easily verifiable tasks, while in the remaining rounds (computation rounds), the actual target computation is delegated. If a high fraction of the outputs of the test rounds are correct, the client can use a majority vote on the outputs of the computation rounds to guaranty that the server did not cheat (except with a negligible probability). The described \emph{cut-and-choose technique}, i.e., intertwining verifiable test rounds with computation rounds, has been used and studied extensively in classical cryptography, e.g. in \cite{rabin1978digitalized,brassard1988minimum,Crepeau2005,CC}. It therefore appeared very promising at first for verifying blind quantum computations as well, and has been used in many previous studies \cite{kapourniotis2023asymmetric, Leichtle,QC_CC,HowToVerif}. However, verification of quantum computation, especially with a potentially malicious server, is challenging if the output is a quantum state. In fact, all previous protocols with quantum output that achieve verifiability with negligible cheating probability, albeit not in the number of rounds but in other parameters such as the code distance, require additional techniques such as error correction \cite{Kashefi_2017_Petros, Mor14, Fitzsimons_2017}. Moreover, a recent work \cite{UnifVDQC_EC} establishes that error correction is, in fact, necessary for verifiable delegated quantum computation based on trap techniques in measurement-based quantum computing. These observations naturally lead to the question whether the cut-and-choose technique alone can suffice for verifying delegated quantum computation or whether more resource-intensive methods are fundamentally unavoidable.

We demonstrate that the latter holds: we prove that every protocol that solely relies on the cut-and-choose technique is affected by a fundamental trade-off between efficiency, correctness, and security. We show that a) for a fidelity-based stand-alone and b) for a generic composable security definition, if a protocol utilizes the cut-and-choose technique with on average $N$ test rounds, the following holds (informally):

\begin{quote}
If $\varepsilon_H$ is the probability of wrongly rejecting the output state, and $\varepsilon_D$ is the deviation (measured with the fidelity or trace distance) of the (malicious) server from an honest behavior that is not noticed by the client, then the sum of the two is lower bounded by a non-negligible function in $N$. 
\end{quote}

We specifically construct an attack that results in these lower bounds for $\varepsilon_H+\varepsilon_D$ for composable and stand-alone security.
Importantly, this attack is separable, independent and identically distributed (i.i.d.) and does not require a predetermined round number or knowledge of the server about the distribution that governs the round number. Hence, our work surpasses intuitive arguments about the limitations of the cut-and-choose technique and provides a formal no-go result for a broad class of protocols.

\section{Background}\label{sec:Results}
In this section, we present the required background and describe the general class of protocols under consideration. We assume a protocol with a probabilistic total of $n+1$ rounds,  of which $n$ are devoted to verification\footnote{Note that there is no loss of generality in assuming only a single state for the output, as outputting more states can only decrease the fidelity if the server is dishonest.}. We assume that the client and the server use a (blind) delegated quantum computation (DQC) subroutine, which allows the client to delegate computations to the server. We consider:
        \begin{itemize}
            \item a probability distribution $\Omega:\mathbb{N}\rightarrow[0,1]$, that the client uses to sample the number of test rounds $n$,
            \item $k$ qubit number of the input and output of the computation,
            \item a family of trap-generators $\{\Theta_{k,n}\}_{k\in\mathbb{N}_+}^{n\in\mathbb{N}}$, each returning for a given round $i$, a test computation and input, i.e., a unitary on $T_i\in\unitary(\Cc^{2^k})$ and a pure state $\chi_i\in \sphere(\Cc^{2^k})$,\footnote{We lift the implied restrictions in the appendix.} 
            \item a family of (positive semidefinite) measurement operators\\ $\left\{\mu_{k,n}(0)\in \pos\left(\Cc^{2^{kn}}\right)\right\}_{k\in\mathbb{N}_+,n\in\mathbb{N}}$ that the client uses to decide whether to accept or reject the server's behavior. 
    \end{itemize}

    \begin{algorithm}[h]
        \caption{The client inputs a state $\psi_C\in \mathcal{D}(\Cc^{2^k})$ and a unitary $U\in\unitary(\Cc^{2^k})$ and receives as output the state $\rho_c$.
        }\label{proto}
        \begin{algorithmic}[1]
            \Statex
            \State The client samples $n\gets_{\Omega}\mathbb{N}$.
            \State The client samples randomly the round that is used for output $\ell\gets_{\$}\{1,...,n+1\}$.
            \For{$i=1,...,n+1$}
                \If{$i\neq \ell$}
                    \State The client delegates a trap-computation $T_i,\chi_i\gets \Theta_{k,n}(i)$ and saves the output of the computation.
                \Else
                    \State The client delegates $U$ on the input $\psi_C$ and saves the output as $\phi_{\ell}$.
                \EndIf
            \EndFor
            \State The client measures the outputs of the trap-computations using $\{\mu_{k,n}(0),\id_{\Cc^{2^k}}^{\otimes n}-\mu_{k,n}(0)\}$. The output is $r$.
            \If{$r=0$}
                \State $\rho_c\gets \phi_\ell$
            \Else
                \State $\rho_c\gets \bot$
            \EndIf
        \end{algorithmic}
    \end{algorithm}

    We prove that cut-and-choose verifiable delegated quantum computing (VDQC) with quantum output cannot be secure and correct, and at the same time efficient, without additional correction or detection techniques, such as error-detection codes. 
Since we prove a no-go result, it suffices to present a successful malicious strategy for one pair of input and computation. 

Let's consider the state $\psi = \ketbra{+}^{\otimes k}$ and the computation $U=\id_2^{\otimes k}$ for an arbitrary $k>0$. 
The server's attack is fairly simple; the server acts honestly in every round, except for applying the operation:
$$P(\alpha) = \left(\begin{matrix}
            1 & 0 \\
           0 & e^{i\alpha}
        \end{matrix}\right)$$ to the $k^{th}$ qubit either before or after the delegated unitary -- note that all currently known implementations of DQC offer at least one of these two types of access to the server. For a given number of test rounds $n$ and output round number $\ell$, the probability that the client accepts when the server performs this attack is:
\begin{align*}
    p_{(n,\ell)}^D = \braket{\mu_{k,n}(0)}{\notimes T^{\alpha}_i\left(\ketbra{\chi_i}\right)}.
\end{align*}
with either $T_i^{\alpha} = T_i\circ (\id_2^{\otimes k-1}\otimes P(\alpha))$ or $T_i^{\alpha} = (\id_2^{\otimes k-1}\otimes P(\alpha))\circ T_i$, depending on the order that $P(\alpha)$ and the delegated unitary were applied. \\

Similarly, for a given $n$ and $\ell$, the acceptance probability $p_{(n,\ell)}^H$ when the server is honest is:
\begin{align*}
    p_{(n,\ell)}^H = \braket{\mu_{k,n}(0)}{\notimes T_i\left(\ketbra{\chi_i}\right)}.
\end{align*}
Hence, the overall probabilities for acceptance if the server is honest ($p_H$), and respectively dishonest ($p_D$), are given by:
\begin{align*}
    p_H = \sum_{n=0}^{\infty}\frac{\Omega(n)}{n+1}\sum_{\ell=1}^{n+1}p_{(n,\ell)}^H, \quad p_D = \sum_{n=0}^{\infty}\frac{\Omega(n)}{n+1}\sum_{\ell=1}^{n+1}p_{(n,\ell)}^D.
\end{align*}
We can derive an upper bound for $|p_H-p_D|$ by using the triangle inequality:
\begin{align}\label{eq:phminpd1}
    |p_H-p_D| &=\left|\sum_{n=0}^{\infty}\frac{\Omega(n)}{n+1}\sum_{\ell=1}^{n+1}\left(p_{(n,\ell)}^H - p_{(n,\ell)}^D\right)\right|
    \leq\sum_{n=0}^{\infty}\frac{\Omega(n)}{n+1}\sum_{\ell=1}^{n+1}\left|p_{(n,\ell)}^H - p_{(n,\ell)}^D\right|.
\end{align}
We can further upper-bound the individual terms by using the fact that the trace distance of $n$-fold tensor products $\notimes T_i(\ketbra{\chi_i})$ and $\notimes T_i^{\alpha}(\ketbra{\chi_i})$ of pure states is given by 
\begin{align*}
    &\left\|\notimes T_i(\ketbra{\chi_i})-\notimes T_i^\alpha(\ketbra{\chi_i})\right\|_1 =2\sqrt{1-\noprod\left|\bra{\chi_i}T_i\D T_i^{\alpha}\ket{\chi_i}\right|^2}.
\end{align*}
Using the Holevo-Helstrom theorem \cite{watrous_2018}, which states that the distinguishability between two states is upper bounded by half the trace distance, we find: 
\begin{align}\label{eq:phminpd2}
    &|p_{(n,\ell)}^H-p_{(n,\ell)}^D|\leq \sqrt{1-\noprod\left|\bra{\chi_i}T_i\D T_i^{\alpha}\ket{\chi_i}\right|^2}.
\end{align}
If the attack is applied before the delegated unitary, we find
\begin{align*}
    \left|\bra{\chi_i}T_i\D T_i^{\alpha}\ket{\chi_i}\right|^2 = \left|\bra{\chi_i}T_i\D T_i \left(\id_2^{\otimes k-1}\otimes P(\alpha)\right)\ket{\chi_i}\right|^2 = \left|\bra{\chi_i} \left(\id_2^{\otimes k-1}\otimes P(\alpha)\right)\ket{\chi_i}\right|^2.
\end{align*}
If, however, the attack happened after the delegated unitary, we find
\begin{align*}
    \left|\bra{\chi_i}T_i\D T_i^{\alpha}\ket{\chi_i}\right|^2 = \left|\bra{\chi_i}T_i\D \left(\id_2^{\otimes k-1}\otimes P(\alpha)\right)T_i\ket{\chi_i}\right|^2.
\end{align*}
So either way it holds
\begin{align*}
    \left|\bra{\chi_i}T_i\D T_i^{\alpha}\ket{\chi_i}\right|^2 \geq \min_{\ket{u}\in \sphere(\Cc^{2^k})}\left|\bra{u}\left(\id_2^{\otimes k-1}\otimes P(\alpha)\right)\ket{u}\right|^2.
\end{align*}
The right-hand side of this inequality is the minimum of the squared absolute values of the elements in the numerical range of $\id_2^{\otimes k-1}\otimes P(\alpha)$. Since this operator has only two different eigenvalues, $1$ and $e^{i\alpha}$, and the numerical range is convex by the Toeplitz–Hausdorff theorem \cite{watrous_2018} we can write
\begin{align*}
    \left|\bra{\chi_i}T_i\D T_i^{\alpha}\ket{\chi_i}\right|^2 \geq\min_{\lambda\in[0,1]}\left|\lambda + (1-\lambda)e^{i\alpha}\right|^2 = \min_{\lambda\in[0,1]}\lambda^2+(1-\lambda)^2+2\lambda(1-\lambda)\cos(\alpha).
\end{align*}
The minimum is obtained at $\lambda = \nicefrac{1}{2}$ which implies
\begin{align*}
    \left|\bra{\chi_i}T_i\D T_i^{\alpha}\ket{\chi_i}\right|^2 &\geq \cos(\frac{\alpha}{2})^2
\end{align*}
Inserted in \eqref{eq:phminpd2} we find
\begin{align*}
    &|p_{(n,\ell)}^H-p_{(n,\ell)}^D|\leq\sqrt{1-\cos(\frac{\alpha}{2})^{2n}},
\end{align*}
which implies for \eqref{eq:phminpd1} 
\begin{align*}
    |p_H-p_D|&\leq\sum_{n=0}^{\infty}\frac{\Omega(n)}{n+1}\sum_{\ell=1}^{n+1}\sqrt{1-\cos(\frac{\alpha}{2})^{2n}} \\&= \sum_{n=0}^{\infty}\Omega(n)\sqrt{1-\cos(\frac{\alpha}{2})^{2n}}
\end{align*}
The function $\sqrt{1-\cos(\nicefrac{\alpha}{2})^{2n}}$ is concave in $n$ (since $0\leq \cos(\nicefrac{\alpha}{2})^{2} \leq 1$), and we can therefore utilize Jensen's inequality and find:
\begin{align}\label{eq:prep_end}
    &|p_H-p_D|\leq \sqrt{1-\cos(\frac{\alpha}{2})^{2N}}=\sqrt{1-\left(1-\sin(\frac{\alpha}{2})^{2}\right)^N}.
\end{align}
where $N$ is the expected number of test rounds. With this preparation, we present our main results in the next section.

\section{Security}
We are now ready to examine the security aspects of the type of protocols presented above. We will first define stand-alone security and prove a trade-off between efficiency and security. We will also prove a similar trade-off for generic composable security definitions by providing a bound for trace distances that are crucial in all composable security frameworks.   

\subsection{Stand-alone security}
    We first define fidelity-based stand-alone security inspired by a similar security definition for quantum state verification used in \cite{colisson2024graphstateverificationprotocols}. The protocol of the client (respectively server) is denoted as $\pi_C$ (respectively  $\pi_S$).
    \begin{definition}\label{def:SASec}
        A protocol $\pi=(\pi_C,\pi_S)$ applied on a resource $\R$ is an $\varepsilon_H$-correct $\varepsilon_D$-secure implementation of VDQC if for any input $\psi\in \mathcal{D}(\X)$ and any unitary $U\in\unitary(\X)$:
        \begin{itemize}
            \item \emph{Correctness}: If the server is honest, it holds\footnote{We use $F(\rho,\sigma)=\left(\Tr{\sqrt{\sqrt{\rho}\sigma\sqrt{\rho}}}\right)^2$.}:
            \begin{align}
                F(\rho_H,U(\psi))\geq 1-\varepsilon_H,
            \end{align}
            where $\rho_H$ is the state obtained from $\pi_C\R\pi_S$.
            \item \emph{Security}: For any attack of the server, it holds:
            \begin{align}
                \max_{p\in[0,1]}&F(\rho_D,p U(\psi) + (1-p)\ketbra{\bot})\geq 1-\varepsilon_D,
            \end{align}
            where $\rho_D$ is the state the client obtains from $\pi_C\R$ in composition with the attack of the server and $\ketbra{\bot}$ is orthogonal to all possible outputs of the protocol in case of no abort.
        \end{itemize}
    \end{definition}
    
\noindent We now prove the trade-off with respect to Definition \ref{def:SASec}, using the results from the previous section. 

\begin{theorem}\label{theo:FidBased}
Let $\pi=(\pi_C,\pi_S)$ be a protocol that implements VDQC from quantum communication channels as described in Protocol \ref{proto}. If $\pi$ is $\varepsilon_H$-correct and $\varepsilon_D$-secure according to Definition \ref{def:SASec}, it holds that $\varepsilon_H+\varepsilon_D\geq\nicefrac{1}{7N}$, where $N$ is the expected number of test rounds. 
    \end{theorem}
    \begin{proof}
        We consider $\psi=\ketbra{+}^{\otimes k}$,  $U=\id_2^{\otimes k}$ and the same attack as in the preparation. The state outputted by the protocol if the server is honest is
        \begin{align}\label{eq:rhoH}
            \rho_H = p_H \ketbra{+}^{\otimes k} + (1-p_H)\ketbra{\bot}. 
        \end{align}
        This implies directly
        \begin{align}
            F\left(\rho_H,\ketbra{+}^{\otimes k}\right) = p_H\Rightarrow \varepsilon_H\geq 1-p_H.
        \end{align}
        If the server is dishonest, we find
        \begin{align}\label{eq:rhoD}
            \rho_D = &p_D \ketbra{+}^{\otimes k-1} \otimes \ketbra{+_{\alpha}} \nnn +&(1-p_D)\ketbra{\bot}. 
        \end{align}
        where $\ket{+_{\alpha}}=\left(\begin{matrix}
            \nicefrac{1}{\sqrt{2}}\\\nicefrac{e^{i\alpha}}{\sqrt{2}}
       \end{matrix}\right)$.\\
        Now, if every pair of positive semi-definite operators $(R_1,S_2)\in \{P_1,Q_1\}\times\{P_2,Q_2\}$ is orthogonal, it holds:
        \begin{align*}
          F(P_1+P_2,Q_1+Q_2) = \left(\sqrt{F(P_1,Q_1)}+\sqrt{F(P_2,Q_2)}\right)^2,
        \end{align*}
        and therefore we can compute the fidelity of $\rho_D$ with $\sigma_p = p\ketbra{+}^{\otimes k}+(1-p)\ketbra{\bot}$:
        \begin{align*}
            F\left(\rho_D,\sigma_p\right) = \left(\sqrt{p_Dp\cos\left(\frac{\alpha}{2}\right)^2}+\sqrt{\left(1-p_D\right)(1-p)}\right)^2.
        \end{align*}
        This implies according to Definition \ref{def:SASec}
        \begin{align*}
            &\varepsilon_D\geq 1-\max_{p\in[0,1]}\left(\sqrt{p_Dp\cos\left(\frac{\alpha}{2}\right)^2}+\sqrt{\left(1-p_D\right)(1-p)}\right)^2,
        \end{align*}
        Since, $\max_{p\in[0,1]}\left(\sqrt{p}a+\sqrt{1-p}b\right)^2 = a^2+b^2$ we find
        \begin{align*}
            \varepsilon_D\geq 1-p_D\cos(\frac{\alpha}{2})^2-(1-p_D) = p_D\sin(\frac{\alpha}{2})^2.
        \end{align*}
        Now we get a lower bound for the sum  of $\varepsilon_H$ and $\varepsilon_D$:
        \begin{align*}
            \varepsilon_H + \varepsilon_D &\geq 1-p_H + p_D \sin\left(\frac{\alpha}{2}\right)^2 \\&\geq \sin\left(\frac{\alpha}{2}\right)^2 (1-|p_H-p_D|).
        \end{align*}
        Using (\ref{eq:prep_end}) next gives us
        \begin{align*}
            \varepsilon_H + \varepsilon_D &\geq\sin\left(\frac{\alpha}{2}\right)^2 \left(1-\sqrt{1-\left(1-\sin(\frac{\alpha}{2})^2\right)^N}\right).
        \end{align*}
        The attacker can choose $\alpha$ and therefore also $\sin(\nicefrac{\alpha}{2})$. We let the attacker choose $\alpha$ such that $\sin(\nicefrac{\alpha}{2})^2=\nicefrac{4}{9N}$ and using $(1-\nicefrac{x}{n})^n\geq 1-x$ for $|x|\leq n$ we find
        \begin{align*}
            \varepsilon_H + \varepsilon_D &\geq\frac{4}{9N}\left(1-\sqrt{1-\left(1-\frac{4}{9}\right)}\right)\nnn
            &=\frac{4}{27N}\geq \frac{1}{7N}
        \end{align*}\qed
    \end{proof}
    \\[1em]
    \subsection{Composable security}
    Composable security definitions in frameworks such as \emph{universal composability} \cite{UC_Canetti}, \emph{abstract cryptography} \cite{MR11}, or \emph{categorical composable cryptography} \cite{Broadbent_2022} fulfill reasonable stand-alone security definitions. So, one could extend the stand-alone no-go result to composable security. However, we note that a more tailored analysis provides a better bound; in fact, the scaling for composable security is $\Omega(\nicefrac{1}{\sqrt{N}})$ and not just $\Omega(\nicefrac{1}{N})$. 
    
    \begin{figure}[h]
    \centering
    \begin{tikzpicture}
        \draw (0,0) rectangle (5,3) node[pos=0.5, align=left] {$\rho = \begin{cases}U(\psi)\text{ if } c=0\\
            \bot \text{ if c=1}\end{cases}$};
        \draw[-{stealth}] (-1,2) -- (0,2) node[pos=0.5,above] {$\psi, U$};
        \draw[-{stealth}] (0,1) -- (-1,1) node[pos=0.5,above] {$\rho$};
        \draw[gray,-{stealth}] (5,2) -- (6,2) node[pos=0.5,above] {$l^{\psi}$};
        \draw[gray,{stealth}-] (5,1) -- (6,1) node[pos=0.5,above] {$c$};
        \node[anchor=south] at (2.5,3) {$\S^{V}$};
        \draw (6,0) rectangle (7,3) node[pos=0.5] {$c=0$};
        \node[anchor=south] at (6.5,3) {$\sharp_S$};
    \end{tikzpicture}
    \caption{Ideal functionality for VDQC. The interface of the server (right interface) is obstructed by a filter $\sharp_S$ if the server is honest, which does not forward $l^{\psi}$ and inputs $c=0$. $l^{\psi}$ is the allowed leakage that contains the register size of $\psi$ and an upper bound for the circuit length of $U$.}
    \label{fig:idealVDQC}
\end{figure}
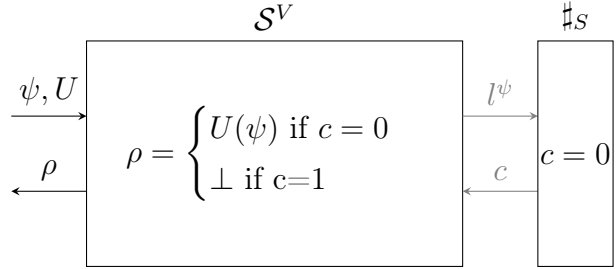
 The canonical ideal resource for VDQC is depicted in Fig. \ref{fig:idealVDQC}. A protocol that is $\varepsilon_H$-correct  with respect to $\S^V\sharp_S$ fulfills:
\begin{align}\label{eq:TDHon}
    \frac{1}{2}\left\|\rho_H-U(\psi)\right\|_1\leq \varepsilon_H
\end{align}
where $\rho_H$ is the client's average output state for the input $\psi$ in the protocol. This is because the trace distance gives the distinguishability according to the Holevo-Helstrom theorem \cite{watrous_2018}, the latter being the relevant quantity in composable cryptography.

If, on the other side, the protocol is $\varepsilon_D$ composably secure with respect to $\S^V$, this implies that for every attack of the server on the implementation, there is an attack on the ideal setup that renders both setups indistinguishable. In \emph{abstract cryptography} and \emph{universal composability}, a simulator attached to the ideal resource converts attacks on the real setting into attacks on the ideal setting; in \emph{categorical composable cryptography}, one defines an attack model and can choose the attack on the ideal setting freely. In either framework, it holds that 
\begin{enumerate}
    \item tracing out the state the server might hold cannot increase distinguishability,
    \item the attack on the ideal setting can only input $c=0$ with some probability $p$ and $c=1$ with probability $1-p$.
\end{enumerate}
Hence, $\varepsilon_D$-composable security implies
\begin{align}\label{eq:TDDis}
    \frac{1}{2}\left\|\rho_D-(pU(\psi)+(1-p)\ketbra{\bot})\right\|_1\leq \varepsilon_D,
\end{align}
where $\rho_D$ is the state of the client in the implementation.
We now prove a lower bound concerning the trace distances presented in \eqref{eq:TDHon} and \eqref{eq:TDDis}. Again, the proof follows ideas similar to those presented in \cite{us}.
\begin{theorem}\label{thm:compo}
    Let $\pi=(\pi_C,\pi_S)$ be a protocol that implements VDQC from quantum communication channels as described in Protocol \ref{proto}, $\rho_H$ be the output state if the server is honest and $\rho_D$ if the server is dishonest. If
    \begin{align*}
        &\frac{1}{2}\left\|\rho_H-U(\psi)\right\|_1\leq \varepsilon_H \text{~~and }\\
        \min_{p\in[0,1]}&\frac{1}{2}\left\|\rho_D-(pU(\psi)+(1-p)\ketbra{\bot})\right\|_1\leq \varepsilon_D,
    \end{align*}
    it holds
    \begin{align*}
        \varepsilon_H+\varepsilon_D\geq \frac{1}{4\sqrt{N}},
    \end{align*}
    where $N$ is the expected number of test rounds.
\end{theorem}
\begin{proof}
    As in Theorem \ref{theo:FidBased}, we consider $\psi=\ketbra{+}^{\otimes k}$ and $U = \id_2^{\otimes k}$. We also consider the same attack, which means that $\rho_H$ and $\rho_D$ are given in \eqref{eq:rhoH} and \eqref{eq:rhoD} respectively. 
    As $\ketbra{\bot}$ is orthogonal to the space of possible output, \eqref{eq:TDHon} becomes:
    \begin{align*}
        \varepsilon_H\geq \frac{1}{2}\left\|\rho_H-\ketbra{+}^{\otimes k}\right\|_1 = 1-p_H
    \end{align*}
    We note that $\|\phi\otimes \nu - \chi\otimes \nu\|_1 = \|\phi - \chi\|_1$ for any density matrices $\phi,\chi$ and $\nu$. Further, if every pair of positive semi-definite operators $(R_1,S_2)\in \{P_1,Q_1\}\times\{P_2,Q_2\}$ is orthogonal, it holds
        \begin{align*}
          \|P_1+P_2-Q_1+Q_2\|_1 = \|P_1-Q_1\|_1 + \|P_2-Q_2\|_1.
    \end{align*}
    Hence, if the server is dishonest, we find for \eqref{eq:TDDis}:
    \begin{align*}
        \varepsilon_D\geq \min_{p\in[0,1]}&\frac{1}{2}\left\|\rho_D-(p\ketbra{+}^{\otimes k}+(1-p)\ketbra{\bot})\right\|_1\\
        =\min_{p\in[0,1]}&\frac{1}{2}\left\|p_D\ketbra{+_\alpha}-p\ketbra{+}\right\|_1 + \frac{1}{2}|p-p_D|.
    \end{align*}
    Similar to \cite{us}, we use the triangle-inequality and the fact that $\|\ketbra{+}\|_1=1$ implies $|p-p_D| = \|p\ketbra{+}-p_D\ketbra{+}\|_1$, and we find:
    \begin{align*}
        \min_{p\in[0,1]}&\frac{1}{2}\left\|p_D\ketbra{+_\alpha}-p\ketbra{+}\right\|_1 + \frac{1}{2}|p-p_D|\geq \frac{p_D}{2}\left\|\ketbra{+_\alpha}-\ketbra{+}\right\|_1.
    \end{align*}
    On the other side, considering $p=p_D$ yields
    \begin{align*}
        \min_{p\in[0,1]}&\frac{1}{2}\left\|p_D\ketbra{+_\alpha}-p\ketbra{+}\right\|_1 + \frac{1}{2}|p-p_D|\leq \frac{p_D}{2}\left\|\ketbra{+_\alpha}-\ketbra{+}\right\|_1,
    \end{align*}
    therefore, the left- and right-side of the above two inequalities are equal. Hence, we find:
    \begin{align*}
        \varepsilon_H+\varepsilon_D&\geq 1-p_H + \frac{p_D}{2}\left\|\ketbra{+_\alpha}-\ketbra{+}\right\|_1 \\&\geq \frac{1}{2}\left\|\ketbra{+_\alpha}-\ketbra{+}\right\|_1\left(1-|p_H-p_D|\right). 
    \end{align*}
    Using $\frac{1}{2}\left\|\ketbra{+_\alpha}-\ketbra{+}\right\|_1 = |\sin(\nicefrac{\alpha}{2})|$ and \eqref{eq:prep_end} we find
    \begin{align*}
        \varepsilon_H+\varepsilon_D&\geq|\sin(\nicefrac{\alpha}{2})|\left(1-\sqrt{1-\left(1-\sin(\nicefrac{\alpha}{2})^2\right)^N}\right).
    \end{align*}
    We finally choose $|\sin(\nicefrac{\alpha}{2})| = \nicefrac{1}{2\sqrt{N}}$ and find with $(1-\nicefrac{x}{n})^n\geq 1-x$ for $|x|\leq n$,
    \begin{align*}
        \varepsilon_H+\varepsilon_D&\geq\frac{1}{2\sqrt{N}}\left(1-\sqrt{1-\left(1-\frac{1}{4N}\right)^N}\right) \geq \frac{1}{4\sqrt{N}}
    \end{align*}\qed
\end{proof}
\section{Discussion and open questions\label{sec:discussion}}
In our work, we prove the existence of a fundamental trade-off between correctness, security, and efficiency of verifiable delegated quantum computing protocols that solely rely on the cut-and-choose technique. We derive this trade-off for both fidelity-based stand-alone security and composable security, and observe that composable security is affected more severely. Our results can also be adapted to settings where multiple, potentially untrusted, clients delegate quantum computations to a remote server. For the sake of simplicity, in the main text we consider that the client can only send separable pure states for the test rounds and has to sample the output round uniformly. Both restrictions are lifted in the appendix, where we consider general test strategies and find a similar trade-off.
    
The same trade-off as in the main part has also been shown to affect quantum state verification protocols that are based on cut-and-choose \cite{us}. This similarity raises the question of whether other functionalities for verification using this technique are similarly affected. However, there are protocols \cite{Mor14,Fitzsimons_2017,Kashefi_2017_Petros} for verifiable delegated quantum computing that circumvent this trade-off by leveraging quantum error correction combined with partitioning the resource state for each computation in the cut-and-choose procedure. We expect that one layer of quantum error correction suffices to provide security and fault tolerance, since there is no relevant difference whether the server is malicious or its devices are imperfect. Nevertheless, this double role of error correction imposes a trade-off on its own: More tolerated errors increase the risk of accepting corrupted states, while tolerating fewer errors demands better experimental implementations. The investigation of verification with noisy devices \cite{Gheorghiu_2019}, which also results in non-negligible security, suggests that the double role of error-correction could pose a problem indeed. It remains to investigate under which conditions this trade-off allows for realistic and secure setups. Finally, our work does show that the overheads of techniques such as error correction are inevitable in principle, which motivates further research in this direction.

\section{Ackowledgments}
F. W. and A. P. acknowledge support from the Emmy Noether DFG grant
No. 41829458. This work was funded by the European Union's Horizon Europe research and innovation programme under grant agreement No. 101102140 – QIA Phase 1. This project was financially supported by BERLIN QUANTUM, an initiative endowed by the Innovation Promotion Fund of the city of Berlin.
\section{References}
\printbibliography[
heading=none]
\appendix
\section{General tests}
In \proref{proto}, the client is restricted to use pure states for the test, to keep no register for themselves and to maintain seperability between the rounds. This restriction allows to have proofs of the trade-offs that are rather similar to these for quantum state verification in \cite{us}, in which the client has no input at all. While this restriction is also motivated by the assumed incapabilities that cause the client to use verifiable delegated quantum computing instead of computing the unitary themselves, it leaves a loophole in the trade-off that will be closed in this subsection. The first step towards closing this loophole is the formalisation of the new considered protocol type. The new component of the protocol type is the \emph{channel network}. The most general thing the client can do is to compose the unitaries of the test rounds using sequential and parallel composition as they like. One condition, however, is that the order of delegating these individual tests is preserved.  
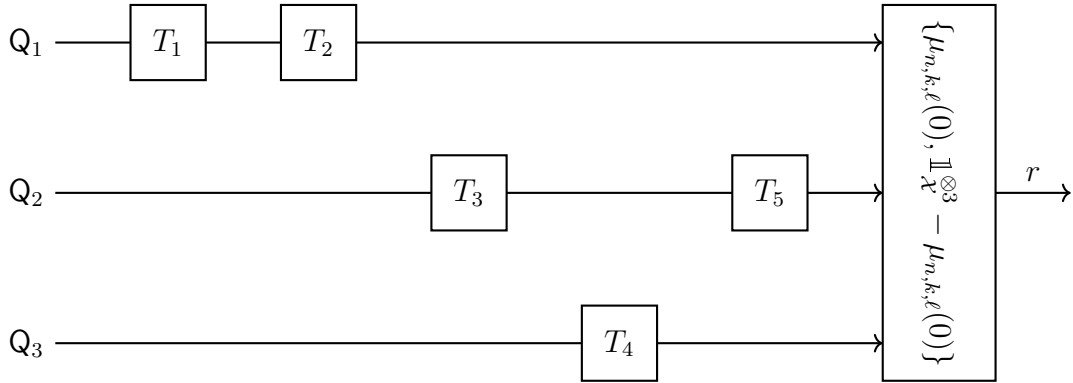
\begin{figure}[H]
    \centering 
    \begin{tikzpicture}
        \draw[thick,->] (0,4.5) -- (11,4.5) node[pos=0,anchor=east] {$\Qsf_1$};
        \draw[thick,->] (0,2.5) -- (11,2.5) node[pos=0,anchor=east] {$\Qsf_2$};
        \draw[thick,->] (0,0.5) -- (11,0.5) node[pos=0,anchor=east] {$\Qsf_3$};
        \draw[thick,fill=white] (1,4) rectangle (2,5) node[pos=0.5] {$T_1$};
        \draw[thick,fill=white] (3,4) rectangle (4,5) node[pos=0.5] {$T_2$};
        \draw[thick,fill=white] (5,2) rectangle (6,3) node[pos=0.5] {$T_3$};
        \draw[thick,fill=white] (7,0) rectangle (8,1) node[pos=0.5] {$T_4$};
        \draw[thick,fill=white] (9,2) rectangle (10,3) node[pos=0.5] {$T_5$};
        \draw[thick] (11,0) rectangle (12.5,5) node[pos=0.5,rotate=-90] {$\{\mu_{n,k,\ell}(0),\id_{\X}^{\otimes 3}-\mu_{n,k,\ell}(0)\}$};
        \draw[thick,->] (12.5,2.5) -- (13.5,2.5) node[pos=0.5,above] {$r$};
    \end{tikzpicture}
    \caption{An example for a network of $5$ test unitaries that the client uses to catch the server cheating followed by the measurement.}
    \label{fig:VDQC_network}
\end{figure}

In \figref{fig:VDQC_network} we see such an example of a channel network, where a client with three registers of size $\density(\X)$ and a register $\Ksf$ they keep for themselves, runs five test rounds. The first register, $\Qsf_1$, is used for the first and second test round, i.e. $T_1\circ T_2$ is applied. The second, $\Qsf_2$, is used for the third and the fifth test, i.e. $T_3\circ T_5$ is applied. Finally, $T_4$ acts on $\Qsf_3$. After this network of unitaries is applied the client measures the registers, including $\Ksf$, and decides whether to accept or reject the server's behavior. 

An important aspect of these networks is the order of the unitaries. Since the client can not change the order of the rounds which label the test unitaries, there is no reasonable interpretation for the application of $T_2$ followed by the application of $T_1$, i.e. $T_2\circ T_1$. This perseverance of order is already captured in \emph{$n$-combs} which formalize these networks, where $n$ is the number of test rounds. An $n$-comb gets a sequence of $n$ operations and intertwines these with $n+1$ other operations called \emph{teeth}, where the teeth are connected by memory registers \cite{Broadbent_2023, us_Combs}. Hence, the $n$-comb is for the client a mapping from $n$ unitaries in $\unitary(\X)$, plugged into the \emph{holes} of the $n$-comb, to a channel of type $\channel\left(\X^{\otimes W},\X^{\otimes W}\right)$, where $W$ depends on the number of registers the client uses. Note that if the client just distributed the unitaries on different registers, the teeth basically just changed which register is part of the memory and which register goes to the next unitary, and the resulting operation will be a unitary. However, in an $n$-comb the client could also apply other operations within the teeth such that the resulting operation is no longer a unitary. 

\subsection{Protocol type}
In the following, the parameters of this general cut-and-choose protocol for verifiable delegated quantum computing will be revisited or added. Lifting the restrictions on the client's test has no influence on formalization of the probability distribution $\Omega:\mathbb{N}\to[0,1]$ used to sample the number of test rounds and also $k$ still denotes the number of qubit register the input consists of and the notation $\X = \Cc^{2^k}$ is used again. The probability distributions $\{\omega_n:\{1,...,n+1\}\to[0,1]\}_{n\in\mathbb{N}}$ govern the round that is used for the output in dependence on $n$. 

The $n$-comb $\mathfrak{B}_{n,k,\ell}:\unitary(\X)^{n}\to\channel\left(\X^{\otimes W_{n,k,\ell}},\X^{\otimes W_{n,k,\ell}}\right)$ that compiles the test operation has to depend on the number of test rounds $n$ and the number of qubit registers $r$. The client could also choose the $n$-comb in dependence of the outcome round $\ell$. The same holds for the number $W_{n,k,\ell}\in \mathbb{N}$ of registers that are used for the tests. 

The test-generators are indeed rather different in the general version; instead of mapping a round to a tuple of a pure quantum state and a unitary, a test $\Theta_{n,k}$ is a tuple of a quantum state $\chi_n\in \density\left(\X^{\otimes W_{n,k,\ell}}\otimes \Y\right)$ and a sequence of unitaries $(T_i\in\unitary(\X))_{i=1}^n$, where $\Y$ is the space that the client uses to boost its distinguishing advantage.

Finally, the measurement also changed; for each test round number $n$, number of qubit register $k$ and output round $\ell$, the client chooses a measurement operator $\mu_{n,k,\ell}(0)\in\pos(\X^{\otimes W_{n,k,\ell}}\otimes \Y)$, which again corresponds to accepting the server's behavior. Note that the register in $\Y$ is measured as well which could allow for catching a cheating server with higher probability.

The protocol type that utilizes this generalized testing is denoted in \proref{protoEnt} below.
\begin{algorithm}[H]
        \caption{The client inputs a register $\Rsf_C$ corresponding to $\density(\X)$ and a unitary $\U\in\unitary(\X)$.
        }\label{protoEnt}
        \begin{algorithmic}[1]
            \Statex
            \State The client samples the number of test rounds $n\gets_{\Omega}\mathbb{N}$ according to $\Omega$.
            \State The client samples the round that is used for output $\ell\gets_{\omega_n}\{1,...,n+1\}$ according to $\omega_n$.
            \State The client fixes $W_{n,k,\ell}$ and $\mathfrak{B}_{n,k,\ell}$.
            \State The client initializes the register $\Qsf_1,...,\Qsf_{W_{n,k,\ell}},\Ksf \gets \chi_n$ with the test state.
            \For{$i=1,...,n+1$}
                \If{$i\neq \ell$}
                    \State The client delegates a test-computation with the input $\Qsf_j$ which is the input register in the $i$th hole of $\mathfrak{B}_{n,k,\ell}$ and the unitary $T_i$ and saves the output of the computation.
                \Else
                    \State The client delegates $\U$ on the input $\Rsf_C$ and loads the result into the register $\Rsf'_C$.
                \EndIf
            \EndFor
            \State The client measures the outputs of the test-computations along with $\Ksf$ using $\left\{\mu_{n,k,\ell}(0),\id_{\X}^{\otimes W_{n,k,\ell}}-\mu_{n,k,\ell}(0)\right\}$. The output is $r$.
            \If{$r=0$}
                \State The client outputs $\Rsf'_C$.
            \Else
                \State The client output $\bot$
            \EndIf
        \end{algorithmic}
    \end{algorithm}

    \subsection{Preliminaries for the proofs of the trade-offs}
    In fact, an attack of the same type as used for the restricted protocol type is sufficient here as well. Hence, the server either applies before or after the delegated unitary a rotation of the form $\id_{\Cc^2}^{\otimes k-1}\otimes P(\alpha)$. For an attack of this type, one can find a similar result as Section \ref{sec:Results} for this general situation. 
    \begin{lemma}\label{lem:prelimVDQCExp}
    Let $\pi=\{\pi_C,\pi_S\}$ be a protocol as described in \proref{proto}. Let 
    \begin{align*}
        p_{(n,k,\ell)}^H&\coloneqq \braket{\mu_{n,k,\ell}(0)}{\left(\mathfrak{B}_{n,k,\ell}((T_i)_{i=1}^{n})\otimes \id_{\Y}\right)(\ketbra{\chi_n})},\\
        p_{(n,k,\ell)}^D&\coloneqq \braket{\mu_{n,k,\ell}(0)}{\left(\mathfrak{B}_{n,k,\ell}((T^{\alpha}_i)_{i=1}^{n})\otimes \id_{\Y}\right)(\ketbra{\chi_n})}
    \end{align*}
    be the acceptance probability for fixed $n$ and $\ell$ for the honest $\left(p_{(n,k,\ell)}^H\right)$ and dishonest $\left(p_{(n,k,\ell)}^D\right)$ setting where the server applies the attack described above. The overall acceptance probabilities are hence given by
    \begin{align*}
        p_H\coloneqq \sum_{n=0}^{\infty}\Omega(n)\sum_{\ell=1}^{n+1}\omega_n(\ell)p_{(n,k,\ell)}^H\mand 
        p_D\coloneqq \sum_{n=0}^{\infty}\Omega(n)\sum_{\ell=1}^{n+1}\omega_n(\ell)p_{(n,k,\ell)}^D.
    \end{align*}
    It holds
    \begin{align*}
        |p_H-p_D|\leq N\left|\sin(\frac{\alpha}{2})\right|,
    \end{align*}
    where $N$ denotes the expected number of test rounds.
\end{lemma}
\begin{proof}
    The proof indeed follows ideas similar to those for the result in Section \ref{sec:Results} and uses standard results for the diamond distance, which we denote as a difference in the completely bounded trace distance $\nicefrac{1}{2}\triplenorm{(\cdot)}_1$, and the trace norm. First, the triangle-inequality and the Holevo-Helstrom theorem imply
    \begin{align*}
        |p_H-p_D| \leq \sum_{n=0}^{\infty}\Omega(n)\sum_{\ell=1}^{n+1}\omega_n(\ell)\frac{1}{2}&\left\|\left(\mathfrak{B}_{n,k,\ell}((T_i)_{i=1}^{n})\otimes \id_{\Y}\right)(\ketbra{\chi_n})-\right.\\&\left.\ \left(\mathfrak{B}_{n,k,\ell}((T^{\alpha}_i)_{i=1}^{n})\otimes \id_{\Y}\right)(\ketbra{\chi_n})\right\|_1.
    \end{align*} 
    Since the diamond distance is larger or equal to the above trace distance (cf. Theorem 3.46 in \cite{watrous_2018}),
    one can write:
    \begin{align*}
        |p_H-p_D| \leq \sum_{n=0}^{\infty}\Omega(n)\sum_{\ell=1}^{n+1}\omega_n(\ell)\frac{1}{2}\triplenorm{\mathfrak{B}_{n,k,\ell}((T_i)_{i=1}^{n})-\mathfrak{B}_{n,k,\ell}((T^{\alpha}_i)_{i=1}^{n})}_1.
    \end{align*}
    Since the diamond distance does not increase under the applications of combs\footnote{We refer here to the discussion in subsection \textit{Distances between networks of channels} in \cite{watrous_2018} or subsection 3.1 in \cite{us_Combs}.}, 
    this implies 
    \begin{align*}
        |p_H-p_D| \leq \sum_{n=0}^{\infty}\Omega(n)\frac{1}{2}\sum_{i=1}^n\triplenorm{T_i-T^{\alpha}_i}_1.
    \end{align*}
    If the server applies the rotation after the delegated unitary, by rewriting the diamond distance in the sum and using the fact that the completely bounded trace norm is multiplicative under the tensor product and sequential composition (cf. Proposition 3.48 and Theorem 3.49 in \cite{watrous_2018}), we get:
    \begin{align*}
        \triplenorm{T_i-T^{\alpha}_i}_1 &= \triplenorm{T_i\left(\id_{\Cc^2}^{\otimes k}-\id_{\Cc^2}^{\otimes k-1}\otimes P(\alpha)\right)}_1 \\&= \triplenorm{T_i}_1 \triplenorm{\id_{\Cc^2}^{\otimes k}-\id_{\Cc^2}^{\otimes k-1}\otimes P(\alpha)}_1 = \triplenorm{\id_{\Cc^2}-P(\alpha)}_1
    \end{align*}
    If the server applies the rotation before the delegated unitary, we get:
    \begin{align*}
        \triplenorm{T_i-T^{\alpha}_i}_1 &= \triplenorm{\left(\id_{\Cc^2}^{\otimes k}-\id_{\Cc^2}^{\otimes k-1}\otimes P(\alpha)\right)T_i}_1 \\&= \triplenorm{\id_{\Cc^2}^{\otimes k}-\id_{\Cc^2}^{\otimes k-1}\otimes P(\alpha)}_1 \triplenorm{T_i}_1 = \triplenorm{\id_{\Cc^2}-P(\alpha)}_1.
    \end{align*}
    Theorem 3.51 in \cite{watrous_2018}) allows us to express the diamond distance as a trace distance of the maps applied on pure states:
    \begin{align*}
        \frac{1}{2}\triplenorm{\id_{\Cc^2}-P(\alpha)}_1 &= \max_{\ket{u}\in\sphere\left(\Cc^{2}\otimes\Cc^{2}\right)}\frac{1}{2}\left\|\id_{\Cc^2\otimes\Cc^2}(\ketbra{u})-(P(\alpha)\otimes \id_{\Cc^2})(\ketbra{u})\right\|_1\\
        & = \sqrt{1-\min_{\ket{u}\in\sphere\left(\Cc^{2}\otimes\Cc^{2}\right)}\left|\bra{u}(P(\alpha)\otimes \id_{\Cc^2})\ket{u}\right|^2},
    \end{align*}
    where we used the expression for the trace distance of pure states for the second equality.
    As in the main part, one finds that
    \begin{align*}
        \min_{\ket{u}\in\sphere\left(\Cc^{2}\otimes\Cc^{2}\right)}\left|\bra{u}(P(\alpha)\otimes \id_{\Cc^2})\ket{u}\right|^2 = \sin(\frac{\alpha}{2}),
    \end{align*}
    since the operator $P(\alpha)\otimes \id_{\Cc^2}$ has only eigenvalues $1$ and $e^{i\alpha}$ and the same arguments for numerical range apply. 
    Hence, using $\frac{1}{2}\triplenorm{T_i-T^{\alpha}_i}_1 = |\sin(\nicefrac{\alpha}{2})|$ for all $0\leq i\leq n$ yields:
    \begin{align*}
        |p_H-p_D| \leq |\sin(\nicefrac{\alpha}{2})|\sum_{n=0}^{\infty}\Omega(n)n = N |\sin(\nicefrac{\alpha}{2})|,
    \end{align*}\qed
\end{proof}

Note how this result differs from the one for the simpler case. The new result uses a bound which essentially is of the same nature as a composition theorem: Distinguishing two compositions with the same structure -- here the same comb -- cannot be harder than distinguishing the components of the compositions. Interestingly, this combination of bounds implies that the overall bound for the difference in acceptance probabilities does not converge for $N\to\infty$ although the difference can be at most $1$ in the first place. 

The result for the simpler case $\left(|p_H-p_D| \leq \sqrt{1-\left(1-\sin(\nicefrac{\alpha}{2})^2\right)^N}\right)$ uses a reduction to an i.i.d. argument. For each round, we considered the optimal measurement for the client,where restricting to pure state testing implies an exponential convergence to $1$.

\subsection{Fidelity-based security}
Equipped with this preparation, the proofs for the trade-offs considering fidelity-based and composable security are not much different from the ones in the simpler setting. First, the trade-off for fidelity-based security is presented in \thmref{theo:FidBased_Ext}.
\begin{theorem}\label{theo:FidBased_Ext}
    Let $\pi=\{\pi_C,\pi_S\}$ be a verifiable delegated quantum computing protocol as described in \proref{protoEnt}. Assume the server is able to apply a unitary on the last qubit register of either the input before the delegated unitary is applied or of the output. Under this condition, if $\pi$ is $\varepsilon_H$-correct and $\varepsilon_D$-secure according to \defref{def:SASec}, it holds that
    \begin{align*}
        \varepsilon_H+\varepsilon_D\geq \frac{1}{7N^2},
    \end{align*}
    where $N$ is the expected number of test rounds.
\end{theorem}
\begin{proof}
    As in \thmref{theo:FidBased}, the input $\psi_C=\ketbra{+}^{\otimes k}$ and the unitary $\U=\id_{\X}$ are considered. By the same argument as for \thmref{theo:FidBased}, one finds
    \begin{align*}
        \varepsilon_H\geq 1-p_H,
    \end{align*}
    where $p_H$ in now given in \lemref{lem:prelimVDQCExp}.

    Also if the server is dishonest, the same manipulations as in \thmref{theo:FidBased} suffice and one finds
    \begin{align*}
        \varepsilon_D\geq p_D(1-\left|\braket{+}{+_{\alpha}}\right|^2),
    \end{align*}
    where $p_H$ in now given in \lemref{lem:prelimVDQCExp}. Using $\left|\braket{+}{+_{\alpha}}\right| = \cos(\nicefrac{\alpha}{2})$ and \lemref{lem:prelimVDQCExp} implies for the sum of $\varepsilon_H$ and $\varepsilon_D$
    \begin{align*}
        \varepsilon_H+\varepsilon_D\geq 1-p_H + p_D\sin(\frac{\alpha}{2})^2 \geq  \sin(\frac{\alpha}{2})^2\left(1-\left|p_H-p_D\right|\right)\geq \sin(\frac{\alpha}{2})^2\left(1-N\left|\sin(\frac{\alpha}{2})\right|\right)
    \end{align*}
    Computing the roots of first derivative after $\sin(\nicefrac{\alpha}{2})$ reveals that $\sin(\nicefrac{\alpha}{2}) = \nicefrac{2}{3N}$ maximizes the lower bounds. This implies
    \begin{align*}
        \varepsilon_H+\varepsilon_D\geq \frac{4}{27N^2} \geq \frac{1}{7N^2},
    \end{align*}
    which proves the claim.\qed
\end{proof}
\subsection{Composable security}
The proof for the trade-off with regard to composable security is also similar to its counterpart for the restricted setting. 
\begin{theorem}
    Let $\pi=\{\pi_C,\pi_S\}$ be a protocol that implements VDQC from quantum communication channels as described in \proref{protoEnt}. If
    \begin{align*}
        &\frac{1}{2}\left\|\rho_H-U(\psi)\right\|_1\leq \varepsilon_H \text{~~and }\\
        \min_{p\in[0,1]}&\frac{1}{2}\left\|\rho_D-(pU(\psi)+(1-p)\ketbra{\bot})\right\|_1\leq \varepsilon_D,
    \end{align*}
    it holds
    \begin{align*}
        \varepsilon_H+\varepsilon_D\geq \frac{1}{4N},
    \end{align*}
    where $N$ is the expected number of test rounds.
\end{theorem}
\begin{proof}
    Again, the input state is $\psi_C=\ketbra{+}^{\otimes k}$ and the delegated unitary is $\U=\id_{\Cc^2}^{\otimes k}$. Similarly as in \thmref{thm:compo}, we find
    \begin{align*}
        \varepsilon_H\geq 1-p_H,
    \end{align*}
    where $p_H$ defined as in \lemref{lem:prelimVDQCExp}.

    If, however, the server is dishonest, by the same argument as in \thmref{thm:compo}, we find
    \begin{align*}
        \varepsilon_D\geq \frac{p_D}{2}\left\|\ketbra{+}^{\otimes k-1}\otimes \ketbra{+_{\alpha}}-\ketbra{+}^{\otimes k}\right\|_1,
    \end{align*}
    where $p_D$ is given as in \lemref{lem:prelimVDQCExp}. Using the expression for the trace distance of pure state, the above becomes
    \begin{align*}
        \varepsilon_D\geq p_D\sqrt{1-\left|\braket{+}{+_{\alpha}}\right|^2} = p_D\sqrt{1-\cos(\frac{\alpha}{2})^2} = p_D \left|\sin(\frac{\alpha}{2})\right|.
    \end{align*}
    Adding $\varepsilon_H$ and $\varepsilon_D$ yields in combination with \lemref{lem:prelimVDQCExp}
    \begin{align*}
        \varepsilon_H + \varepsilon_D &\geq 1-p_H + p_D \left|\sin(\frac{\alpha}{2})\right| \geq \left|\sin(\frac{\alpha}{2})\right|\left(1-\left|p_H-p_D\right|\right) 
        \\&\geq \left|\sin(\frac{\alpha}{2})\right|\left(1-N\left|\sin(\frac{\alpha}{2})\right|\right).
    \end{align*}
    The choice $\sin(\nicefrac{\alpha}{2})=\nicefrac{1}{2N}$ maximizes the lower bound and yields $\varepsilon_H + \varepsilon_D\geq \nicefrac{1}{4N}$.\qed
\end{proof}
\end{document}